\documentclass{llncs}

\usepackage{times}
\usepackage{color}
\usepackage[numbers]{natbib}

\makeatletter 
\renewcommand\bibsection%
{ 
  \section*{\refname 
    \@mkboth{\MakeUppercase{\refname}}{\MakeUppercase{\refname}}}
}           
\makeatother

\usepackage{float}
\usepackage{epsfig}
\usepackage{amsmath}
\usepackage{amsfonts}
\usepackage{amssymb}
\usepackage{amstext}
\usepackage{amsmath}
\usepackage{hyperref}
\usepackage{xspace}
\usepackage{latexsym}
\usepackage{verbatim}
\usepackage{multirow}
\usepackage{ifthen}
\usepackage{url}

\usepackage{graphpap,amscd,mathrsfs,graphicx,lscape}
\usepackage{epsfig, color, subfigure}
\usepackage[ruled]{algorithm2e}
\renewcommand{\algorithmcfname}{ALGORITHM}
\SetAlFnt{\small}
\SetAlCapFnt{\small}
\SetAlCapNameFnt{\small}
\SetAlCapHSkip{0pt}
\IncMargin{-\parindent}

\newcommand{\E}{\mathbb{E}}

\renewcommand{\comment}[1]{}

\newcommand{\val}{v}
\newcommand{\vals}{\mathbf{\val}}

\newcommand{\bid}{b}
\newcommand{\bids}{{\mathbf \bid}}

\begin{document}
\markboth{B. Sivan and V. Syrgkanis}{Vickrey Auctions for Irregular
Distributions}
\title{Vickrey Auctions for Irregular Distributions}
\author{Balasubramanian Sivan\inst{1}
\and
Vasilis Syrgkanis\inst{2}
}

\institute{
Microsoft Research, \email{bsivan@microsoft.com}
\and
Computer Science Dept., Cornell University, \email{vasilis@cs.cornell.edu}
}

\pagestyle{plain}

\iffalse
\terms{Algorithms, Economics, Theory}
\terms{Mechanism design, Revenue maximization}
\acmformat{Balasubramanian Sivan and Vasilis Syrgkanis, 2013. Vickrey Auction
for Irregular Distributions}
\fi

\maketitle{}
\begin{abstract}
The classic result of Bulow and Klemperer~\cite{BK96} says that in a single-item
auction recruiting one more bidder and running the Vickrey auction achieves a
higher revenue than the optimal auction's revenue on the original set of
bidders, when values are drawn i.i.d. from a regular distribution. We give a
version of Bulow and Klemperer's result in settings where bidders' values are
drawn from non-i.i.d. irregular distributions. We do this by modeling irregular
distributions as some convex combination of regular distributions. The regular
distributions that constitute the irregular distribution correspond to different
population groups in the bidder population. Drawing a bidder from this
collection of population groups is equivalent to drawing from some convex combination of these
regular distributions. We show that recruiting one extra bidder from each
underlying population group and running the Vickrey auction gives at least half
of the optimal auction's revenue on the original set of bidders.


\end{abstract}
\section{Introduction}
\label{sec:intro}
% !TEX root = irregular_bulow_klemperer.tex
Simplicity and detail-freeness are two much sought-after themes in auction
design.  The celebrated classic result of Bulow and Klemperer~\cite{BK96} says
that in a standard single-item auction with $n$ bidders, when the valuations of
bidders are drawn  i.i.d from a distribution that satisfies a regularity
condition, running a Vickrey auction (second-price auction) with one extra
bidder drawn from the same distribution yields at least as much revenue as the
optimal auction for the original $n$ bidders. The Vickrey auction is both
simple and detail-free since it doesn't require any knowledge of bidder
distributions. Given this success story for i.i.d. regular distributions, we ask
in this paper, what is the analogous result when we go beyond i.i.d regular
settings?  Our main result is a version of  Bulow and Klemperer's result to
non-i.i.d irregular settings.  Our work gives the first positive
results in designing simple mechanisms for irregular distributions, by parameterizing
irregular distributions, i.e., quantifying the amount of irregularity in a distribution. Our
parameterization is motivated by real world market structures and in turn
indicates that most realistic markets will not be highly irregular with respect
to this metric. Our results enable the first positive approximation bounds on the
revenue of the second-price auction with an anonymous reserve in both i.i.d.
and non-i.i.d. irregular settings.

Before explaining our results, we briefly describe our setting. We consider a
single-item auction setting with bidders having quasi-linear utilities.  That
is the utility of a bidder is his value for the item if he wins, less the price
he is charged by the auction. We study auctions in the Bayesian setting, i.e.
the valuations of bidders are drawn from known distributions\footnote{One of
the goals of this work is to design detail-free mechanisms, i.e., minimize the
dependence on  knowledge of distributions. Thus most of our results make little
or no use of knowledge of distributions. We state our dependence precisely
while stating our results.}.  We make the standard assumption that bidder
valuations are drawn from independent distributions. 
%The standard auction that we use throughout is the Vickrey auction (or the
%second-price auction) possibly with a reserve price. This auction solicits
%sealed bids from agents and chooses the highest bidder who exceeds the reserve
%price and charges him the maximum of the reserve price and the second highest
%value.  Finally, a widely assumed technical condition on bidder distributions
%is the regularity condition.  This states that the function $x -
%\frac{1}{h(x)}$ is non-decreasing, where the function $h(x) =
%\frac{f(x)}{1-F(x)}$ is the hazard-rate function. 

\paragraph{Irregular distributions are common.}
The technical regularity condition in Bulow and Klemperer's result is quite
restrictive, and indeed irregular distributions are quite common in markets.
For instance, any distribution with more than a single mode violates the
regularity condition. 
%\vsedit{Intuitively, the regularity condition
%mathematically corresponds to the marginal revenue curve (derivative of revenue
%with respect to quantity produced) of a market being monotonic (see
%\cite{Bulow1989}). Though the latter is a reasonable assumption for a
%homogeneous market, it is easily broken in the presence of heterogeneity as we
%describe below.}
The most prevalent reason for a bidder's valuation distribution failing to
satisfy the regularity condition is that a bidder in an auction is randomly
drawn from a heterogeneous population. The population typically is composed of
several groups, and each group has its characteristic preferences. For instance
the population might consist of students and seniors, with each group typically
having very different preferences from the other. While the distribution of
preferences within any one group might be relatively well-aligned and the value
distribution might have a single mode and satisfy the regularity condition, the
distribution of a bidder drawn from the general population, which is a mixture
of such groups, is some convex combination of these individual distributions.
Such a convex combination violates regularity even in the simplest cases. 

For a variety of reasons, including legal reasons and absence of good data, a
seller might be unable to discriminate between the buyers from different
population groups and thus has to deal with the market as if each buyer was
arriving from an irregular distribution. However, to the least, most sellers do
know that their market consists of distinct segments with their characteristic
preferences.

\paragraph{Measure of Irregularity.} The above description suggests that a
concrete measure of irregularity of a distribution is the number of regular
distributions required to describe it.  We believe that such a measure could be
of interest in both designing mechanisms and developing good provable revenue
guarantees for irregular distributions in many settings.  It is a rigorous
measure of irregularity for any distribution since any distribution can be
well-approximated almost everywhere by a sufficient number of regular ones and
if we allow the number of regular distributions to grow to infinity then any
distribution can be exactly described\footnote{This follows from the fact that a
uniform distribution over an interval is a regular distribution and every
distribution can be approximated in the limit using just uniform distributions.} Irregular distributions that typically
arise in practice are combinations of a small number of regular distributions
and this number can be considered almost a constant with respect to the market
size. In fact there exist  evidence in recent
\cite{Johnson2003,Guimaraes2011} and classical \cite{Robinson1933}
microeconomic literature  that irregularity of the value distribution
predominantly arises due to market segmentation in a small number of parts
(e.g. loyal customers vs. bargain-hunters \cite{Johnson2003}, luxury vs. low
income buyers \cite{Guimaraes2011} etc). Only highly pathological
distributions require a large number of regular distributions to be described
--- such a setting in a market implies that the population is heavily segmented
and each segment has significantly different preferences from the rest. 

Motivated by this, we consider the following setting: the market/population
consists of $k$ underlying population groups, and the valuation distribution of
each group satisfies the regularity condition. Each bidder is drawn according
to some probability distribution over these groups. That is bidder $i$
arrives from group $t$ with probability $p_{i,t}$. Thus if $F_t$ is the
cumulative distribution function (cdf) of group $t$, the cdf of bidder $i$ is
$F_i = \sum_t p_{i,t} F_t$. For example, consider a market for a product that
consists of two different population groups, say students and seniors. Now
suppose that two bidders come from two cities with different student to senior
ratios. This would lead to the probability $p_{i,t}$'s to be different for
different $i$'s. This places us in a non-i.i.d. irregular setting. All our
results also extend to the case where these probabilities $p_{i,t}$ are
arbitrarily correlated. 

\begin{example}[An illustrative example]
\label{ex:exp} Consider an eBay seller of an
ipad. One could think of the market as segmented mainly in two groups of
buyers: young and elder audience. These two market segments have completely
different value distributions. Suppose for instance, that the value
distribution of young people is distributed as a normal distribution
$N(\mu_y,\sigma)$ while the elder's is distributed as a normal distribution
$N(\mu_e,\sigma)$ with $\mu_y> \mu_e$.  In addition, suppose that the eBay
buyer population is composed of a fraction $p_y$ young people and $p_e<p_y$ of
elders.  Thus the eBay seller is facing an irregular valuation distribution
that is a mixture of two Gaussian distribution with mixture probabilities
$p_y$ and $p_e$ (see Figure \ref{fig:irreg-example}). Even more generally, this mixture
could be dependent on the city of the buyer and hence be different for different buyers.

\begin{figure}\label{fig:irreg-example}
\subfigure{
\includegraphics[scale=.7]{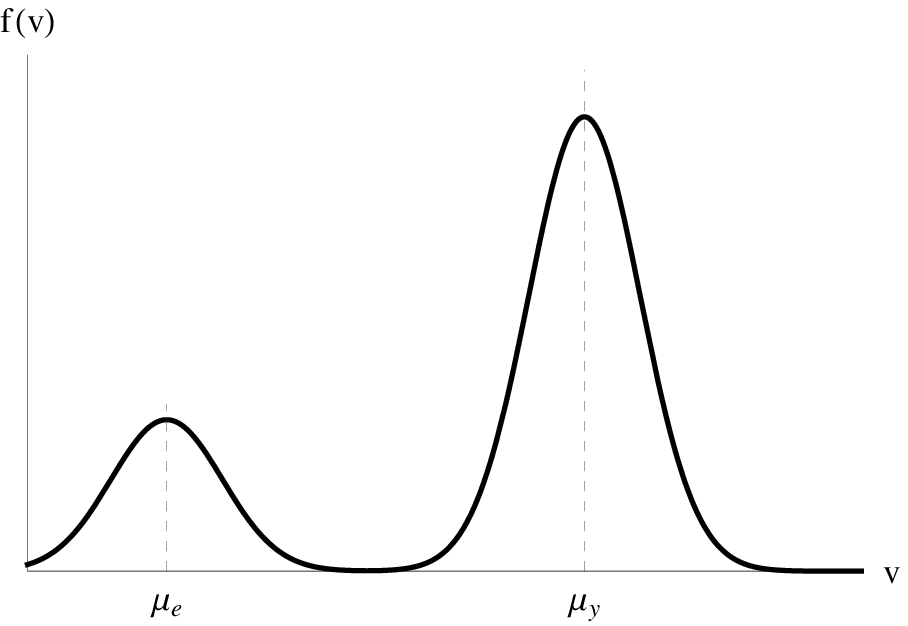}}
\subfigure{
\includegraphics[scale=.7]{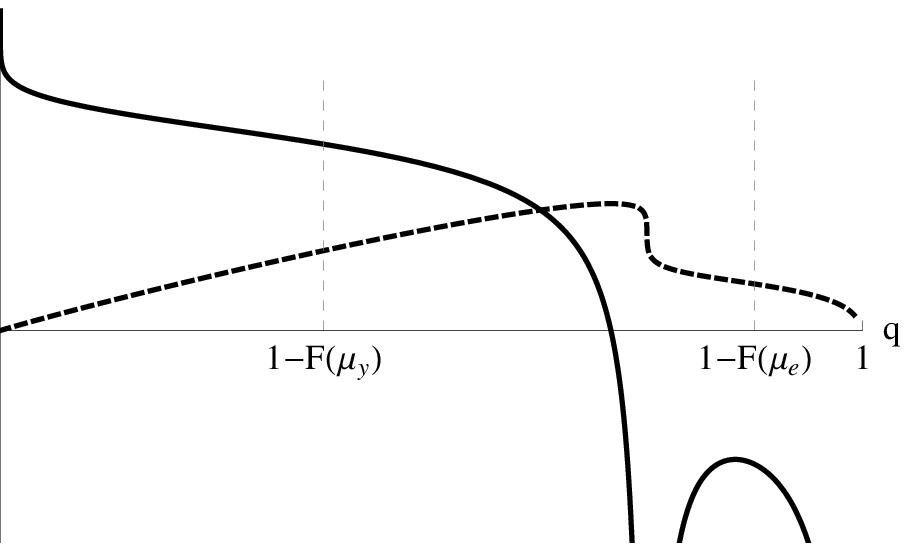}}
\caption{Left figure depicts pdf of the bimodal distribution of valuations of Example \ref{ex:exp}, while the right
figure depicts the revenue (dashed) $R(q) = q\cdot F^{-1}(1-q)$ where $q$ is the probability of sale, and the
marginal revenue curve $\frac{dR(q)}{dq}$ for this distribution.}
\end{figure}

The eBay seller has two ways of increasing the revenue that he receives:
a) increasing the market competition by bringing extra bidders through advertising
(possibly even targeted advertising), and b) setting appropriately his reserve price in
the second price auction that he runs. Observe that he has no means of price
discriminating. Even running Myerson's auction which is non-discriminatory for
i.i.d. settings leads to randomization. In particular, randomization leads to the
undesirable feature of sometimes serving an agent with smaller value. This raises the main questions that we
address in this paper: how should he run his advertising campaign? How many
people more (either through targeted or non-targeted advertising) should he
bring to the auction to get a good approximation to the optimal revenue? What
approximation of the optimal revenue is he guaranteed by running a Vickrey
auction with a single anonymous reserve? 

%\vsdelete{None of the existing literature in approximately optimal auctions provides any
%positive results for irregular distributions, and hence, would be unable to
%give any provable guarantees for these questions, even in this simple example.
%On the contrary,} 
Giving a sneak preview of our main results, our paper gives
positive results to all the above questions: 1) bringing just one extra young
bidder in the auction (targeted advertising) and running a Vickrey auction with
no reserve would yield revenue at least $1/2$ of the optimal revenue (Theorem
\ref{thm:hazard-rate}), 2) bringing $2$ extra bidders drawn from the
distribution of the combined population (non-targeted advertising) would yield
at least $\frac{1}{2}\left(1-\frac{1}{e}\right)$ of the optimal revenue
(Theorem \ref{thm:iid-hazard-rate}), 3) running a Vickrey auction among the
original $n$ bidders with an anonymous reserve price can yield an $8$-approximation of the optimal
revenue (Theorem \ref{thm:anonIrregularNiid}). 
\end{example}

\subsection*{Our Results} 

%What is the analogue of Bulow and Klemperer's result for such settings?
\paragraph{First result (Section~\ref{sec:irregular}): Targeted Advertising for non-i.i.d. irregular settings.}
We show that by recruiting an extra bidder from each underlying group and
running the Vickrey auction, we get a revenue that is at least half of the
optimal auction's revenue in the original setting. While the optimal auction
is manifestly impractical in a non-i.i.d. irregular setting due to its
complicated rules, delicate dependence on knowledge of distribution and its
discriminatory nature\footnote{The optimal auction in a non-i.i.d. setting will
award the item to the bidder with the highest virtual value and this is not
necessarily the bidder with the highest value.  In addition, typically a
different reserve price will be set to different bidders.  This kind of
discrimination is often illegal or impractical. Also, the exact form of the
irregular distribution will determine which region's of bidder valuations will
need to be "ironed", i.e.  treated equally. %Thus a very small error in the modeling of the distribution can lead to very suboptimal revenue.
}, the Vickrey
auction with extra bidders is simple and detail-free: it makes no use of the
distributions of bidders.  The auctioneer must just be able to identify that
his market is composed of different groups and must conduct a targeted
advertising campaign to recruit one extra bidder from each group. 
%An interesting open question that we leave for future work would be to
%identify the number amount of extra people from each population group needed
%to reach optimal revenue. 
This result can be interpreted as follows: while advertising was the solution proposed by
Bulow and Klemperer~\cite{BK96} for i.i.d. regular distributions,
\textit{targeted advertising} is the right approach for non-i.i.d. irregular
distributions. 
%\vscomment{deleted example from here, moved it to section of main theorem.}
%

\paragraph{Tightness.} While we do not know if the the factor $2$ approximation we get is tight,
Hartline and Roughgarden~\cite{HR09} show that even in a non-i.i.d. regular
setting with just two bidders it is impossible to get better than a
$4/3$-approximation by duplicating the bidders, i.e., recruiting $n$ more
bidders distributed identically to the original $n$ bidders. This lower bound
clearly carries over to our setting also: there are instances where recruiting only one
bidder from each different population group cannot give anything better than a 
$4/3$-approximation.

\paragraph{Second result (Main result, Section~\ref{sec:dominant}): Just one
extra bidder for hazard rate dominant distributions.} If the $k$ underlying
distributions are such that one of them stochastically dominates, hazard-rate
wise, the rest, then we show that recruiting just one extra bidder from the
hazard rate dominant distribution and running the Vickrey auction gets at least
half of the optimal revenue for the original setting.  A distribution $F$ hazard
rate dominates a distribution $G$ iff for every $x$ in the intersection of the
support of the two distributions the hazard rate $h_F(x) (=
\frac{f(x)}{1-F(x)}$, where $f(\cdot)$ and $F(\cdot)$ are the pdf and cdf
respectively) is at most $h_G(x) (= \frac{g(x)}{1-G(x)}$, where $g(\cdot)$ and
$G(\cdot)$ are the pdf and cdf
respectively). We denote
such a domination by $F \succeq_{hr} G$.

Further, hazard rate dominance requirement is not uncommon: for instance, if all
the $k$ underlying distributions were from the same family of distributions like
the uniform, exponential, Gaussian or even power law, then one of them is
guaranteed to hazard rate dominate the rest.  
%For example if the $k$ underlying distributions
%were uniform distributions $U[a_i,b_i]$, then the distribution with the highest
%upper bound $b_i$ hazard rate dominates the rest of them; if all the $k$
%underlying distributions were exponential, or Gaussian distributions
%with the same variance, the distribution with the highest mean hazard rate
%dominates the rest. If the underlying distributions are power-law
%distributions, the distribution with the smallest scaling exponent hazard-rate
%dominates the rest (See Section~\ref{sec:dominant} for a proof). 
%In such situations, recruiting \textit{only one} extra bidder from the
%stochastically dominant distribution is enough to guarantee a half
%approximation. 
Though several common distributions satisfy this hazard rate
dominance property, it has never been
previously exploited in the context of approximately optimal auctions.

\paragraph{Third result (Section~\ref{sec:iid-irregular}): Non-targeted
advertising for i.i.d.  irregular distributions.} When the bidders are
identically distributed, i.e., the probability $p_{i,t}$ of distribution $t$ getting
picked for bidder $i$ is the same for all $i$ (say $p_t$), we show that 
%a natural question 
%to ask is how many extra bidders should we bring from the original distribution
if each $p_t \geq \delta$, then  
bringing $\Theta\left(\frac{log(k)}{\delta}\right)$ extra bidders drawn from the
original distribution (and not from one of the $k$ underlying distributions) yields a constant approximation to the optimal revenue. 
%This result follows from our first result. 
Further in the special case
where one of the underlying regular distributions hazard rate dominates the rest and
its mixture probability is $\delta$ then $\Theta\left(\frac{1}{\delta}\right)$ bidders drawn
from the original distribution are enough to yield a constant approximation.
This shows that when each of the underlying population groups is sufficiently 
thick, then recruiting a few extra bidders from the original distribution
is all that is necessary.
% to approximate the optimal auction. 

\begin{remark}
For the latter result it is not necessary that the decomposition
of the irregular distribution that we use should resemble
the actual underlying population groups. Even if the market is highly fragmented
with several population groups, as long as there is mathematically some way to decompose
the irregular distribution into the convex combination of a few regular distributions
our third result still holds. 
\end{remark}

\paragraph{Fourth result (Section~\ref{sec:singleReserve}): Vickrey with a Single
(Anonymous) Reserve.} Suppose we are unable to recruit extra bidders. What is
the next simplest non-discriminatory auction we could hope for? The Vickrey
auction with a single reserve price. We show that when the non-i.i.d irregular distributions
all arise as different convex combinations of the $k$ underlying regular
distributions, there exists a reserve such that the Vickrey auction with this
single reserve obtains a $4k$ approximation to the optimal revenue. Though
the factor of approximation is not small, it is the first non-trivial
approximation known for non-i.i.d irregular distributions via Vickrey with
anonymous reserve. In addition, as we
already explained, in typical market applications we expect the number of
different population groups $k$ to be some small constant.
% and therefore the
%latter approximation bound will be reasonable for natural scenarios.

What is the best one can hope for a non-i.i.d irregular setting? Chawla,
Hartline and Kleinberg~\cite{CHK07} show that for general non-i.i.d irregular
distributions it is impossible to get a $o(\log n)$ approximation using Vickrey
auction with a single reserve price, and it is unknown if this lower bound is
tight, i.e., we do not yet know of a $\Theta(\log n)$ approximation.  However
the $o(\log n)$ impossibility exists only for arbitrary non-i.i.d irregular
settings and doesn't apply when you assume some natural structure on the
irregularity of the distributions, which is what we do.  

To put our results in context: Single reserve price Vickrey auctions were also
analyzed by Hartline and Roughgarden~\cite{HR09} for non-i.i.d \textit{regular}
settings, that showed that there exists a single reserve price that obtains a
$4$-approximation. Chawla et al.~\cite{CHMS10} show that when bidders are drawn
from non-i.i.d irregular distributions, a Vickrey auction with a
distribution-specific reserve price obtains a $2$-approximation. Thus if there
are $k$ different distributions, $k$ different reserve prices are used in this
result. This means that if we insist on placing a single (anonymous) reserve
price, this result guarantees a $O(1/k)$ approximation. In particular, when all
distributions are different, i.e. $k=n$, this boils down to a $O(1/n)$
approximation. 

In contrast, our result shows that even when all the distributions are
different, as long as every irregular distribution can be described as some
convex combination of $k$ regular distributions, Vickrey with a single reserve
price gives a factor $4k$ approximation.  Further the factor does not grow with
the number of players $n$. 
%Finally, we believe that there is scope for
%improvement in this factor of $4k$ and leave it as an open question. 

\begin{remark}
We also show that if the bidders are distributed with identical mixtures and the mixture
probability is at least $\delta$ then Vickrey auction with a single reserve
achieves a $\Theta\left(1+\frac{\log(k)}{n\delta}\right)$ approximation. If one of the
regular distribution hazard rate dominates the rest and has mixture probability
$\delta$, then Vickrey with a single reserve achieves a
$\Theta\left(1+\frac{1}{n\delta}\right)$ approximation.
\end{remark}

Observe that if all $k$ regular
distributions in the mixture have equal probability of arriving, then our results
shows that a Vickrey auction with a single reserve achieves at least a
$\Theta\left(1+\frac{k\log(k)}{n}\right)$ of the optimal revenue. This
approximation ratio becomes better as the number of bidders increases, as long
as the number of underlying regular distributions remains fixed. If
the number of underlying distributions increases linearly with the number of
bidders, then the result implies a $\Theta(\log(n))$ approximation,
matching the lower bound of \cite{CHMS10}.

\paragraph{Related Work.} Studying the trade-off between simple and
optimal auctions has been a topic of interest for long in auction design. The
most famous result is the already discussed result of Bulow and
Klemperer~\cite{BK96} for single-item auctions in i.i.d regular settings.
Hartline and Roughgarden~\cite{HR09} generalize~\cite{BK96}'s result for
settings beyond single-item auctions: they consider auctions where the set of
buyers who can be simultaneously served form the independent set of a
matroid; further they also relax the i.i.d constraint and deal with non-i.i.d
settings.  Dhangwatnotai, Roughgarden and Yan~\cite{DRY10} study revenue
approximations via VCG mechanisms with multiple reserve prices, where the
reserve prices are obtained by using the valuations of bidders as a sample from
the distributions.  Their results apply for matroidal settings when the
distributions are regular, and for general downward closed settings when the
distributions satisfy the more restrictive monotone hazard rate condition. 
%Single-sample mechanisms are not so relevant with the single-item setting that
%we study here, since in a single item setting a Vickrey auction always yields
%higher revenue than a single sample auction. In addition, we focus more on
%bidder augmentation results of the form of Bulow and Klemperer rather than
%prior-free mechanisms. In addition our results on imply better approximation
%guarantee's for the revenue of the Vickrey auction than those implied by
%\cite{DRY10}. 
As previously discussed, Chawla et al.~\cite{CHMS10} show that for i.i.d
irregular distributions, Vickrey auction with a single reserve price gives a
2-approximation to the optimal revenue and for non-i.i.d distributions Vickrey
auction with a distribution-specific reserve price guarantees a
$2$-approximation; Chawla et al.~\cite{CHK07} show that it is impossible to
achieve a $o(\log n)$ approximation via Vickrey auction with a single reserve
price for non-i.i.d irregular distributions. Single-item Vickrey auctions with bidder
specific monopoly reserve prices were also studied in
Neeman~\cite{N03} and Ronen~\cite{R01}. Approximate revenue maximization via
VCG mechanisms with supply limitations were studied in Devanur et
al.~\cite{DHKN11} and Roughgarden et al.~\cite{RTY12}.

\section{Preliminaries}
\label{sec:prelim}

%\section{Model and Notation}
\paragraph{Basic model.}
We study single item auctions among $n$ bidders. Bidder $i$ has a value $v_i$ for a good, and the valuation profile for all
the $n$ players together is denoted by $\vals = (\val_1,\val_2,\dots,\val_n)$. In a sealed bid auction each player submits
a bid, and the bid profile is denoted by $\bids = (b_1,b_2,\dots, b_n)$. An auction is a pair of functions $(\mathbf{x},\mathbf{p})$,
where $x$ maps a bid vector to outcomes $\{0,1\}^n$, and $p$ maps a bid vector to $\mathbf {R}_+^n$, i.e., a 
non-negative payment for each player. The players have quasi-linear utility functions, i.e., their utilities have a separable and
linear dependence on money, given by $u_i(v_i, \vals_{-i}) = v_ix_i(\vals) - p_i(\vals)$. An auction is said to be dominant strategy truthful if submitting
a bid equal to your value yields no smaller utility than any other bid in every situation, i.e., for all 
$\vals_{-i}$, $v_ix_i(\vals)-p_i(\vals) \geq v_ix_i(b_i,\vals_{-i})-p_i(b_i,\vals_{-i})$. Since we focus on truthful auctions in this paper $\bids$ = $\vals$
from now on. 

\paragraph{Distributions.}
We study auctions in a Bayesian setting, i.e., the valuations of bidders are drawn from a distribution. In particular, we assume that valuation of bidder $i$
is drawn from distribution $F_i$, which is independent from but not necessarily identical to $F_j$ for $j\neq i$. For ease of presentation, we assume that these distributions
are continuous, i.e., they have density function $f_i$. We assume that the support of these distributions are intervals on the non-negative real line, with non-zero density 
everywhere in the interval. 

\paragraph{Regularity and irregularity.}
The hazard rate function of a distribution is defined as $h(x) = \frac{f(x)}{1-F(x)}$. A distribution is said to have a Monotone Hazard Rate(MHR)
if $h(x)$ is monotonically non-decreasing. A weaker requirement on distributions is called regularity: the function $\phi(x) = x - \frac{1}{h(x)}$
is monotonically non-decreasing. We do not assume either of these technical conditions for our distributions. Instead we assume that the market of bidders 
consists of $k$ groups and each group has a regular distribution $G_i$ over valuations.  Each bidder is drawn according to some (potentially different) convex combination
of these $k$ regular distributions, i.e., $F_i(x) = \sum_{t=1}^{k} p_{i,t}G_t(x)$. Such a distribution $F_i(\cdot)$ in most cases significantly violates the regularity condition.

%Almost all irregular distributions used as examples in the literature (eg.
%Sydney Opera house distribution) can be written as a convex combination of a
%small number of regular distributions. 
In fact, mathematically, any irregular distribution can be approximated by a
convex combination of sufficiently many regular distributions and as we take the
number of regular distributions to infinity then it can be described exactly.
Thus the number of regular distributions needed to describe an irregular
distribution is a valid measure of irregularity that is well-defined for any
distribution.

%\paragraph{\textbf{Auction formats used.}}
%The \textit{Vickrey auction} or the second-price auction awards the item to the agent with the highest bid and charges him the second highest value.
%The \textit{Vickrey auction with a reserve} computes the highest bidder and awards the item to him only if his bid exceeds the reserve price, and charges him the 
%maximum of the second highest bid and the reserve price. The \textit{Vickrey auction with $k$ extra bidders}, is just the Vickrey auction run on $n+k$ bidders instead of the original $n$ bidders. Both the Vickrey auction and the Vickrey auction with reserve do not need any knowledge of distribution. The Vickrey auction with a reserve
%requires just a single parameter from the distribution, namely the reserve price itself. 

\paragraph{Revenue Objective.}
The objective in this paper to design auctions to maximize expected revenue, i.e., the expectation of the sum of the payments of all agents. Formally, the objective is to 
maximize $\E_{\vals}[\sum_i p_i(\vals)]$. Myerson~\cite{M81} characterized the expected revenue from any auction as its expected virtual surplus, i.e. the expected sum of virtual values
of the agents who receive the item, where the virtual value of an agent is $\phi(v) = v - \frac{1}{h(v)}$. 
Formally, for all bidders $i$, $\E_{\vals}[p_i(\vals)] = \E_{\vals}[\phi_i(v_i)x_i(\vals)]$. The equality holds even if we condition on a fixed $v_{-i}$, i.e.,
$\E_{v_i}[p_i(v_i,\vals_{-i})] = \E_{v_i}[\phi(v_i)x_i(v_i,\vals_{-i})]$.

%\section{Non-i.i.d Regular Settings}
%\label{sec:regular}
%\input{regular.tex}

\section{Targeted Advertising and the Non-i.i.d. Irregular Setting}
\label{sec:irregular}
% !TEX root=irregular_bulow_klemperer.tex
In this section we give our version of Bulow and Klemperer's result~\cite{BK96}
for non-i.i.d irregular distributions. 

\begin{theorem}
\label{thm:mainThm}
Consider an auction among $n$ non-i.i.d irregular bidders where each bidder's distribution $F_i$ is some
mixture of $k$ regular distributions $\{G_1,\ldots,G_k\}$ (the set of regular distributions is 
the same for all bidders but the mixture probabilities could be different). The revenue 
of the optimal auction in this setting is at most twice the revenue of a Vickrey
auction with $k$ extra bidders, where each bidder is drawn from a distinct distribution from $\{G_1,\dots,G_k\}$.  
\end{theorem}
\begin{proof}
Bidder $i$'s distribution $F_i(x) = \sum_{t=1}^k p_{i,t} G_t(x)$ can be thought of
as being drawn based on the following process: first a biased $k$-valued coin is
flipped that decides from which distribution $G_t$ player $i$'s value will come
from (according to the probabilities $p_{i,t}$), and then a sample from $G_t$ is
drawn. Likewise, the entire valuation profile can be thought of as being drawn
in a similar way: first $n$ independent, and possibly non-identically biased,
$k$-valued coin tosses, decide the regular distribution from each bidder's value
is going to be drawn from.  Subsequently a sample is drawn from each
distribution. 

Let the random variable $q_i$ be the index of the regular distribution that bidder $i$'s value is going to be drawn, i.e., 
$q_i$ is the result of the coin toss for bidder $i$. Let $q$ denote the index profile 
of all players. Let $p(q)=\prod_{i=1}^n p_{i,q_i}$ be the probability that the index profile $q$ results
after the $n$ coin tosses. Let $G(q)=\times_i G_{q_i}$ be the joint product distribution of players' values conditioned on 
the profile being $q$.  

Let $M_q$ be the optimal mechanism when bidders' distribution profile is $q$. 
Let $\mathcal{R}_{M}^q$ be the expected revenue of mechanism $M_q$. Let $R_{M}^q(\vals)$ denote the 
revenue of the mechanism when bidders have value $\vals$. The revenue of the optimal mechanism $M$ which cannot learn and exploit the actual 
distribution profile $q$ is upper bounded by the revenue of the optimal mechanism that can first learn $q$. Therefore we have,
\begin{equation}
\mathcal{R}_M\leq\sum_{q\in [1..k]^n} p(q) \E_{\vals\sim G(q)}[R_M^q(\vals)]
\end{equation}

Now, $\E_{\vals\sim G(q)}[R_M^q(\vals)]$ corresponds to the optimal expected revenue when bidder $i$'s distribution is the regular distribution $G_{q_i}$. Let $k(q)$ denote the number of distinct regular distributions contained in the profile $q$. Note that $k(q) \leq k$ for all $q$. Thus the above expectation corresponds to the revenue of a single-item auction where players can 
be categorized in $k(q)$ groups and bidders within each group $t$ are distributed i.i.d. according to a regular distribution $G_t$. Theorem 6.3 of \cite{RTY12} applies to such a setting and shows
that the optimal revenue for each of these non-i.i.d regular settings will be at 
most twice the revenue of Vickrey auction with one extra bidder for each distinct distribution in the profile $q$. 
Hence,
\begin{align}
\mathcal{R}_M\leq\sum_{q\in [1..k]^n} p(q) \E_{\vals\sim G(q)}[R_M^q(\vals)]&\leq 
2\sum_{q\in [1..k]^n} p(q) \E_{\vals\sim G(q)}[R_{SP_{n+k(q)}}(\vals)]\label{eqn:mainEqn}\\
& \leq 2\sum_{q\in [1..k]^n} p(q) \E_{\vals\sim G(q)}[R_{SP_{n+k}}(\vals)]\label{eqn:mainEqnFinal}
\end{align}
Since, the Vickrey auction with $k$ extra bidders doesn't depend on the index profile $q$
the RHS of~\eqref{eqn:mainEqnFinal} corresponds to the expected revenue of
$SP_{n+k}$ when bidders come from the initial i.i.d irregular distributions. 
\qed\end{proof}

The above proof actually proves an even stronger claim: the revenue from running the  
Vickrey auction with $k$ extra bidders is at least half approximate even if the auctioneer could distinguish bidders by learning the bidder distribution profile $q$ and
run the corresponding optimal auction $R_M^q$. 

\textit{Lower bound.} A corner case of our theorem is when each bidder comes from a different regular
distribution. From Hartline and Roughgarden~\cite{HR09} we know that a lower
bound of $4/3$ exists for such a case. In other words there exists two regular
distributions such that if the initial bidders came each from a different
distribution among these, then adding two extra bidders from those
distributions will not give the optimal revenue but rather a $4/3$
approximation to it.  The same lower bound proves
that if bidders came from the same mixture of these two regular distributions
(i.e. are i.i.d), then the expected revenue of the auction that first
distinguishes from which regular distribution each bidder comes from and then
applies the optimal auction, yields  higher revenue than adding two extra
bidders from the two distributions and running a Vickrey auction.

\section{Just one extra bidder for hazard rate dominant distributions}
\label{sec:dominant}
In this section we examine the setting where among the $k$ underlying regular
distributions there exists one distribution that stochastically dominates the
rest in the sense of hazard rate dominance.  Hazard rate dominance is a
standard dominance concept used while establishing revenue guarantees for
auctions (see for example~\cite{K11}) and states the following: A
distribution $F$ hazard rate dominates a distribution $G$ iff for every $x$ in
the intersection of the support of the two distributions: $h_F(x)\leq h_G(x)$.
We denote such a domination by $F\succeq_{hr} G$.

In such a setting it is natural to ask whether adding just a single player from the dominant distribution
is enough to produce good revenue guarantees. We actually show that adding only one extra person
coming from the dominant distribution achieves exactly the same worst-case guarantee as adding 
$k$ extra bidders one from each underlying distribution. 
%We also make the mild assumption that the supports of the distributions are intervals.

\begin{theorem}
\label{thm:hazard-rate}
 Consider an auction among $n$ non-i.i.d irregular bidders where each bidder's distribution $F_i$ is some mixture of $k$ regular distributions $\{G_1,\ldots,G_k\}$ such that $G_1\succeq_{hr} G_t$ for 
all $t$. The revenue of the optimal auction in this setting is at most twice the revenue of a Vickrey
auction with one extra bidder drawn from $G_1$.
\end{theorem}

The proof is based on a new lemma for the regular distribution setting: bidders are drawn
from a family of $k$ regular distributions such that one of them hazard-rate
dominates the rest. This lemma can be extended to prove
Theorem~\ref{thm:hazard-rate} in a manner identical to how Theorem 6.3 of
Roughgarden et al.~\cite{RTY12} was extended to prove
Theorem~\ref{thm:mainThm} in our paper. We don't repeat that extension here,
and instead just prove the lemma. The lemma uses the notion of commensurate
auctions defined by Hartline and Roughgarden~\cite{HR09}.

\begin{lemma}
\label{lem:hazardCommensurateExtension}
Consider a non-i.i.d. regular setting where each player's value comes from some
set of distributions $\{F_1,\ldots,F_k\}$ such that $F_1\succeq_{hr} F_t$ for
all $t$. The optimal revenue of this setting is at most twice the revenue of
Vickrey auction with one extra bidder drawn from $F_1$. 
\end{lemma}
\begin{proof}
Let $\vals$ denote the valuation profile of the initial $n$ bidders and let $v^*$ the valuation of the 
extra bidder from the dominant distribution. Let $R(\vals,v^*)$ and $S(\vals,v^*)$ denote the winners of the optimal auction ($M$) and of 
the second price auction with the extra bidder ($SP_{n+1}$) respectively. We will show
that the two auctions are commensurate (see \cite{HR09}) which is sufficient for proving the lemma. Establishing commensurateness boils down to showing
that:
\begin{align}
\E_{\vals,v^*}[\phi_{S(\vals,v^*)}(v_{S(\vals,v^*)})|S(\vals,v^*)\neq R(\vals,v^*)]\geq~& 0\label{eqn:virtualdom}\\
\E_{\vals,v^*}[\phi_{R(\vals,v^*)}(v_{R(\vals,v^*)})|S(\vals,v^*)\neq R(\vals,v^*)]\leq~& \E_{\vals,v^*}[p_{S(\vals,v^*)}|S(\vals,v^*)\neq R(\vals,v^*)]\label{eqn:pricedom}
\end{align}
where $p_S$ is the price paid by the winner of the second price auction. 
The proof of equation~\eqref{eqn:pricedom} is easy and very closely 
follows the proof in~\cite{HR09} above.

We now prove equation~\eqref{eqn:virtualdom}. Since $F_1\succeq_{hr} F_t$ we have that for all $x$ in the intersection of the support of $F_1$ and $F_t$: 
$h_1(x)\leq h_t(x)$, which in turn implies that $\phi_1(x)\leq \phi_t(x)$,
since $\phi_t(x)=x-\frac{1}{h_t(x)}$. By the definition of winner in Vickrey auction we have $\forall i: v_{S(\vals,v^*)}\geq v_i$. In particular, $v_{S(\vals,v^*)}\geq v^*$. If $v^*$ is in the support of $F_{S(\vals,v^*)}$, then the latter, by regularity of distributions, implies that $\phi_{S(\vals,v^*)}(v_{S(\vals,v^*)})\geq \phi_{S(\vals,v^*)}(v^*)$. Now $F_1\succeq_{hr}F_t$ implies that $\phi_{S(\vals,v^*)}(v^*)\geq \phi_{1}(v^*)$ (since by definition $v^*$ must
also be in the support of $F_1$). If $v^*$ is not in the support of $F_{S(\vals,v^*)}$, then since $v^*<v_{S(\vals,v^*)}$ and 
all the supports are intervals, it must be that $v^*$ is below the lower bound $L$ of the support of 
$F_{S(\vals,v^*)}$. Wlog we can assume that the support of $F_1$ intersects the support of every other distribution. Hence, since $v^*$ is below
$L$ and the support of $F_1$ is an interval, $L$ will also be in the support of $F_1$. Thus $L$ is
in the intersection of the two supports. By regularity of $F_{S(\vals,v^*)}, F_1$ and by the hazard rate dominance 
assumption, we have $\phi_{S(\vals,v^*)}(v_{S(\vals,v^*)})\geq \phi_{S(\vals,v^*)}(L)\geq \phi_1(L)\geq \phi_1(v^*)$. 
Thus in any case $\phi_{S(\vals,v^*)}(v_{S(\vals,v^*)})\geq \phi_1(v^*)$. Hence, we immediately get that:
\begin{align*}
\E_{\vals,v^*}[\phi_{S(\vals,v^*)}(v_{S(\vals,v^*)})|S(\vals,v^*)\neq R(\vals,v^*)]\geq
\E_{\vals,v^*}[\phi_{1}(v^*)|S(\vals,v^*)\neq R(\vals,v^*)]
\end{align*}
Conditioned on $\vals$ the latter expectation becomes:
\begin{align*}
\E_{v^*}[\phi_1(v^*)|S(\vals,v^*)\neq R(\vals,v^*),\vals]
\end{align*}
But conditioned on $\vals$, $R(\vals,v^*)$ is some fixed bidder $i$. Hence, the latter expectation is
equivalent to: $\E_{v^*}[\phi_1(v^*)|S(\vals,v^*)\neq i]$ for some $i$. We claim that for all $i$ the latter expectation must be positive. Conditioned
on $\vals$, the event $S(\vals,v^*)\neq i$ happens only if $v^*$ is sufficiently high, i.e., there is a threshold $\theta(\vals)$ such that $S(\vals,v^*) \neq i$ happens only if $v^*\geq \theta(\vals)$
(if $i$ was the maximum valued bidder in the profile $\vals$ then $\theta(\vals) = v_i$, else $\theta(\vals) = 0$.)
By regularity of distributions, $v^* \geq \theta(\vals)$ translates to $\phi_1(v^*)\geq \phi_1(\theta)$.
So we now have to show that: $\E_{v^*}[\phi_1(v^*)|\phi_1(v^*)\geq \phi_1(\theta)] \geq 0$.
Since the unconditional expectation of virtual value is already non-negative, the expectation conditioned on a lower bound on virtual values is clearly non-negative.
\qed\end{proof}

\paragraph{Examples and Applications.} There are many situations where a
hazard-rate dominant distribution aactually exists in the market. We provide
some examples below. 

%\textit{Uniform distributions.}
%Suppose the $k$ underlying distributions were all uniform distributions of the form $U[a_i,b_i]$. The hazard rate $h_i(x) = \frac{f_i(x)}{1-F_i(x)} = \frac{\frac{1}{b_i-a_i}}{\frac{b_i-x}{b_i-a_i}} = \frac{1}{b_i - x}$. Clearly, the distribution with a larger $b_i$ hazard-rate dominates the distribution with a larger $b_i$. Thus there will always exist a hazard-rate dominant distribution. 
%
%\textit{Exponential  distributions.}
%If the $k$ underlying distributions were all exponential distributions, i.e., $G_i(x) = 1-e^{-\lambda_ix}$, then the hazard rate $h_i(x) = \lambda_i$. Thus the distribution with the smallest $\lambda_i$ hazard rate dominates the rest. 
%
%\textit{Power-law distributions.}
%If the $k$ underlying distributions were all power-law distributions, namely, $G_i(x) = 1-\frac{1}{x^{\alpha_i}}$, then the hazard rate $h_i(x) - \frac{\alpha_i}{x}$. Thus the distribution with the smallest $\alpha_i$ hazard-rate dominates the rest. 

\textit{Uniform, Exponential, Power-law distributions.}
Suppose the $k$ underlying distributions were all uniform distributions of the
form $U[a_i,b_i]$. The hazard rate $h_i(x) = \frac{1}{b_i - x}$. Clearly, the
distribution with a larger $b_i$ hazard-rate dominates the distribution with a 
smaller $b_i$. If the $k$ underlying distributions were all exponential distributions, i.e., $G_i(x) = 1-e^{-\lambda_ix}$, then the hazard rate $h_i(x) = \lambda_i$. Thus the distribution with the smallest $\lambda_i$ hazard rate dominates the rest. 
If the $k$ underlying distributions were all power-law distributions, namely, $G_i(x) = 1-\frac{1}{x^{\alpha_i}}$, then the hazard rate $h_i(x) = \frac{\alpha_i}{x}$. Thus the distribution with the smallest $\alpha_i$ hazard-rate dominates the rest. 

\textit{A general condition. }
If all the $k$ underlying regular distributions were such that for any pair $i$, $j$ they satisfy $1-G_i(x) = (1-G_j(x))^{\theta_{ij}}$, then it is easy to verify that there always exists one distribution that hazard-rate dominates the rest of the distributions. For instance, the family of exponential distributions, and the family of power-law distributions are special cases of this general condition.

\section{Non-Targeted Advertising and the i.i.d. Irregular Setting}
\label{sec:iid-irregular}
% !TEX root=irregular_bulow_klemperer.tex
In this section we consider the setting where all the bidders are drawn from
the same distribution $F$.  We assume that $F$ can be written as a convex
combination of $k$ regular distributions $F_1,\ldots,F_k$, i.e. $F =
\sum_{t=1}^k p_t F_t$ and such that the mixture probability $p_t$ for every
distribution is at least some constant $\delta$: $\forall t\in [1,\ldots,k]:
p_t\geq \delta$. A natural question to ask in an i.i.d. setting is how many
extra bidders should be recruited from the original distribution to achieve a
constant fraction of the optimal revenue (i.e., by running a non-targeted
advertising campaign)?

In this section answer the above question as a function of the number of
underlying distributions $k$ and the minimum mixture probability $\delta$. We
remark that our results in this section don't require the decomposition of $F$
into the $F_t$'s resemble the distribution of the underlying population groups.
Even if the number of underlying population groups is very large, as long as
there is some mathematical way of decomposing $F$ into $k$ regular
distributions with a minimum mixture probability of $\delta$, our results go
through. Hence, one can optimize our result for each $F$ by finding the
decomposition that minimizes our approximation ratio.

%The following two theorems are proved in Apendix~\ref{app:iid-irregular}
\begin{theorem}
\label{thm:iid-irregular}
Consider an auction among $n$ i.i.d. irregular bidders where the bidders'
distribution $F$ can be decomposed into a mixture of $k$ regular distributions
$\{G_1,\ldots,G_k\}$ with minimum mixture probability $\delta$.  The revenue of
the optimal auction in this setting is at most $2\frac{k+1}{k}$ the revenue of
a Vickrey auction with $\Theta\left(\frac{log(k)}{\delta}\right)$ extra bidders
drawn from distribution $F$.
\end{theorem}
\begin{proof}
Suppose that we bring $n^*$ extra bidders in the auction. Even if the
decomposition of the distribution $F$ doesn't correspond to an actual market
decomposition, we can always think of the value of each of the bidders drawn as
follows: first we draw a number $t$ from $1$ to $k$ according to the mixture
probabilities $p_t$ and then we draw a value from distribution $G_t$.

Let $\mathcal{E}$ be the event that all numbers $1$ to $k$ are represented by
the $n^*$ random numbers drawn to produce the value of the $n^*$ extra bidders.
The problem is a generalization of the coupon collector problem where there are
$k$ coupons and each coupon arrives with probability $p_t\geq \delta$. The
relevant question is, what is the probability that all the coupons are collected
after $n^*$ coupon draws? The probability that a coupon $t$ is not collected
after $n^*$ draws is: $(1-p_t)^{n^*}\leq (1-\delta)^{n^*}\leq e^{-n^*\delta}$.
Hence, by the union bound, the probability that some coupon is not collected
after $n^*$ draws is at most $k e^{-n^*\delta}$. Thus the probability of event
$\mathcal{E}$ is at least $1-ke^{-n^*\delta}$. Thus if
$n^*=\frac{\log(k)+\log(k+1)}{\delta}$ then the probability of $\mathcal{E}$ is
at least $1-\frac{1}{k+1}$.

Conditional on event $\mathcal{E}$ happening we know that the revenue of the
auction is the revenue of the initial auction with at least one player extra
drawn from each of the underlying $k$ regular distributions.  Thus we can apply
our main theorem \ref{thm:mainThm} to get that the expected revenue conditional
on $\mathcal{E}$ is at least $\frac{1}{2}$ of the optimal revenue with only the
initial $n$ bidders. Thus: \begin{align*} \mathcal{R}_{SP_{n+n^*}} \geq
\left(1-\frac{1}{k+1}\right)\E_{v,\tilde{v}\sim
F^{n+n^*}}[R_{SP_{n+n^*}}(v,\tilde{v})|\mathcal{E}] \geq
\left(1-\frac{1}{k+1}\right)\frac{1}{2}\mathcal{R}_M
\end{align*}
\qed\end{proof}

\begin{theorem}
\label{thm:iid-hazard-rate}
Consider an auction among $n$ i.i.d. irregular bidders where the bidders' distribution $F$ can be decomposed
into a mixture of $k$ regular distributions $\{G_1,\ldots,G_k\}$ such that $G_1$ hazard rate dominates $G_t$
for all $t>1$. The revenue of the optimal auction in this setting is at most $2\frac{e}{e-1}$ the revenue of a Vickrey
auction with $\frac{1}{p_1}$ extra bidders drawn from distribution $F$.
\end{theorem} 
\begin{proof}
Similar to theorem \ref{thm:iid-irregular} conditional on the even that an extra player is drawn from
the hazard rate distribution, we can apply Lemma \ref{lem:hazardCommensurateExtension} to get that this
conditional expected revenue is at least half the optimal revenue with the initial set of players. 
If we bring $n^*$ extra players then the probability of the above event happening is $1-(1-p_1)^{n^*}\geq 1-e^{-n^*p_1}$. Setting $n^*=\frac{1}{p_1}$ we get the theorem.
\qed\end{proof}

\paragraph{Prior-Independent Mechanisms.} The two theorems above imply
prior-independent mechanisms for the i.i.d. irregular setting based on a
reasoning similar to the one used by~\cite{DRY10} in converting Bulow-Klemperer
results to prior-independent mechanims in the i.i.d. regular setting.
Specifically, instead of bringing $k$ extra i.i.d. bidders we could use the
maximum value of a random subset of $k$ existing bidders as a reserve on the
remaining $n-k$ bidders. The theorems above then imply that this
prior-independent mechanism yields a constant approximation with respect to the
optimal mechanism among the $n-k$ bidders. Further, since the bidders are all
i.i.d., and the $k$ bidders were chosen before their valuation are drawn, the
expected optimal revenue among the $n-k$ bidders is at least $1-\frac{k}{n}$ of
the optimal revenue among the $n$ bidders. Thus as long as the number of bidders
$k$ required by Theorems \ref{thm:iid-irregular} and \ref{thm:iid-hazard-rate}
is smaller than $n$, this approach yields a prior-independent mechanism with a
meaningful revenue approximation guarantee.  Hence, Theorem
\ref{thm:iid-irregular} implies that if $n\delta\geq c\log(k)$ (i.e. the
expected number of players from each population is at least $c\log(k)$) the
random sampling mechanism described above is $2\frac{k+1}{k}\frac{c}{c-1}$
approximate.  Similarly, Theorem \ref{thm:iid-hazard-rate} implies that it is
$2\frac{e}{e-1}\frac{c}{c-1}$-approximate, if $n\cdot p_1\geq c$, i.e. the
expected number of players from the hazard-rate dominant distribution at least
$c$.

\section{Vickrey with Single Reserve for Irregular Settings}
\label{sec:singleReserve}
% !TEX root=irregular_bulow_klemperer.tex
In this section we prove revenue guarantees for Vickrey auction with a single
reserve in the general irregular setting. 

\begin{theorem}\label{thm:anonIrregularNiid}
Consider an auction among $n$ non-i.i.d irregular bidders where each bidder's distribution $F_i$ is some
mixture of $k$ regular distribution $\{G_1,\ldots,G_k\}$ (the set of regular distributions is 
the same for all bidders but the mixture probabilities could be different). The revenue 
of the optimal auction in the above setting is at most $4k$ times the revenue of a
second price auction with a single reserve price which corresponds to the monopoly reserve price of one of the $k$ distributions $G_i$.  
%(Proof in Appendix~\ref{app:anonymous})
\end{theorem}
\begin{proof}
We use the same notation as in Section~\ref{sec:irregular}.  In particular, we let $q$ denote the index profile 
of distributions for all players and $p(q)=\prod_{i=1}^n p_{i,{q_i}}$ be the probability that an index profile arises.
Let $G(q)=\times_i G_{q_i}$ be the product distribution that corresponds to
how players values are distributed conditional on the coin tosses having value $q$. 

Let $M_q$ be the optimal mechanism when bidders' distribution profile is $q$. 
Let $\mathcal{R}_{M}^q$ be the expected revenue of mechanism $M_q$. By equation~\eqref{eqn:mainEqn} in Section~\ref{sec:irregular} we have,
\begin{align}
\mathcal{R}_M\leq\sum_{q\in [1..k]^n} p(q) \E_{\vals\sim G(q)}[R_M^q(\vals)]&\leq 
2\sum_{q\in [1..k]^n} p(q) \E_{\vals\sim G(q)}[R_{SP_{n+k(q)}}(\vals)]\label{eqn:mainEqnRepeat}
\end{align}

Consider the auction $SP_{n+k(q)}$. If instead of adding the $k(q)$ extra bidders, we place a random reserve drawn from the distribution of the maximum value among the $k(q)$ extra bidders, and ran the Vickrey auction. Call the later $SP_n(R(q))$. If the winner of the auction $SP_{n+k(q)}$ is one among the original $n$ bidders, clearly  $SP_{n+k(q)}$ and $SP_n(R(q))$ will have the same expected revenue. Further, the expected revenue of $SP_{n+k(q)}$  conditioned on the winner being one among the original $n$ bidders is no smaller than the expected revenue of $SP_{n+k(q)}$ conditioned on the winner being one among the newly added $k(q)$ bidders. Also, the probability that the winner comes from the newly added $k(q)$ bidders is at most $1/2$. Thus $SP_n(R(q)) \geq \frac{1}{2}SP_{n+k(q)}$. Combining this with Equation~\eqref{eqn:mainEqnRepeat}, we have
\begin{align}
\mathcal{R}_M&\leq 2\sum_{q\in [1..k]^n} p(q) \E_{\vals\sim G(q)}[R_{SP_{n+k(q)}}(\vals)]
\leq 4\sum_{q\in [1..k]^n} p(q) \E_{\vals\sim G(q)}[R_{SP_{n}(R(q))}(\vals)]\nonumber\\
%&= 4\sum_{q\in [1..k]^n} p(q) \E_{\vals\sim G(q)}[\sum_{t=1}^{k}R_{SP_{n}(R(q),t)}(\vals)]\\
&= \sum_{t=1}^{k}4\sum_{q\in [1..k]^n} p(q) \E_{\vals\sim G(q)}[R_{SP_{n}(R(q),t)}(\vals)]\label{eqn:groupSplit}\\
&\leq 4k\sum_{q\in [1..k]^n} p(q) \E_{\vals\sim G(q)}[R_{SP_{n}(R(q),t^*)}(\vals)]\label{eqn:bestGroup}
\end{align}
In equation~\eqref{eqn:groupSplit}, the revenue $R_{SP_{n}(R(q))}(\vals)$ is written as $\sum_{t=1}^{k}R_{SP_{n}(R(q),t)}(\vals)$, i.e., as the sum of contributions from each population group. Given this split, there exists a polulation group $t^*$ that gets at least $\frac{1}{k}$ fraction of all groups together, and thus at least $\frac{1}{4k}$ fraction of the optimal mechanism, which is what is expressed through inequality~\eqref{eqn:bestGroup}. 

Now the auction $SP_{n}(R(q))$ from the perspective of the group $t^*$ is just the Vickrey auction run for group $t^*$ alone with a single random reserve of $\max\{R(q)$, Maximum value from groups other than $t^* \}$. However within the group $t^*$ since we are in a i.i.d regular setting it is optimal to run Vickrey auction for the group $t^*$ alone with the monopoly reserve price of that group. That is if we replace the single reserve of $\max\{R(q)$, Maximum value from groups other than $t^* \}$ with the optimal (monopoly) reserve price for $t^*$, Vickrey auction for group $t^*$ with such a reserve gives no lesser revenue, and this holds for every $q$! Finally, when we add in the agents from other groups, single-item Vickrey auction's revenue for the entire population with monopoly reserve price of group $t^*$ is no smaller than the revenue of single-item Vickrey auction for group $t^*$ alone with the monopoly reserve price of group $t^*$. Chaining the last two statements proves the theorem. 
%Note that the distribution chosen for the randomized reserve $R(q)$ depends only on the distinct distributions present in $q$. Thus there are at most $2^k$ such randomized reserves. Hence choosing one among the $2^k$ randomized reserves at random will give us a $\frac{1}{4}\cdot\frac{1}{2^k}$ approximation to $\mathcal{R}_M$. This also means that there is a single deterministic reserve that gets the same approximation. 
\qed\end{proof}

\bibliographystyle{acmsmall}
\bibliography{ibk}
\end{document}